\documentclass[journal]{IEEEtran}
\newtheorem{theorem}{Theorem}
\newtheorem{lemma}{Lemma}

\newtheorem{problem}{Problem}
\newtheorem{proof}{Proof}
\newtheorem{assumption}{Assumption}

\ifCLASSINFOpdf
\else
\fi
\usepackage{amsmath,amssymb,amsfonts, algorithm}
\usepackage{algpseudocode}
\usepackage{graphicx}
\usepackage{enumitem}   

\begin{document}
%
\title{Certified Feedforward Tracking for Unknown Nonlinear Systems via Invertible Neural Networks}
%
%
%
\author{Berk Altiner, Rajasree Sarkar, Arunava Banerjee, 
Zongxuan Sun, and Kenneth Kim%
\thanks{Berk Altiner, Rajasree Sarkar, Arunava Banerjee, and Zongxuan Sun
are with the Department of Mechanical Engineering, University of
Minnesota--Twin Cities, Minneapolis, MN 55455, USA
(e-mail: \{altin009, sarka122, abanerje, zsun\}@umn.edu).}%
\thanks{Kenneth Kim is with DEVCOM Army Research Laboratory,
Aberdeen Proving Ground, Aberdeen, MD 21005, USA
(e-mail: kenneth.s.kim11.civ@army.mil).}%
}
\maketitle

\begin{abstract}
In this paper, we address the certification of data-driven feedforward control for periodic tracking of unknown nonlinear systems under partial state measurements. To this end, we adopt an invertible neural network (INN) as a surrogate for the unknown system. This choice allows us to bypass solving a nonconvex inversion problem, eliminating the associated inversion errors and reducing tracking error certification to a surrogate modeling problem. We then apply conformal prediction to provide finite-sample probabilistic guarantees on the surrogate modeling error which, through the derived tracking error bound, yield marginal certificates on feedforward tracking error. Finally, we demonstrate the approach on a DC-motor-driven mechanical load with nonlinear friction.
\end{abstract}


%
\IEEEpeerreviewmaketitle

\section{Introduction}
Feedforward controllers play a central role in achieving high-precision motion control and are classically designed by inverting a model of the plant \cite{boerlage2003model}, \cite{jamaludin2009friction}. The performance of such controllers is fundamentally limited by model fidelity \cite{devasia2002should}. However, obtaining an accurate model of a complex nonlinear system requires both expert knowledge of the underlying physics and substantial effort in parameter identification, and is often difficult in practice. To circumvent this difficulty, data-driven modeling and control have emerged as a viable alternative, enabling feedforward design for unknown or partially unknown dynamics \cite{xu2024robust}, \cite{kon2022physics}. 

In a data-driven setting, neural networks have become a popular choice for surrogate modeling due to their approximation capacity~\cite{zhang2007inverse, 
wanigasekara2019neural, aarnoudse2021control, aarnoudse2024control}. Although this flexibility is desirable, the risk of overparameterization 
motivates the incorporation of physical prior knowledge into the model structure. Physics-guided neural networks (PGNNs) have received particular attention in this context. In this line, a PGNN in an NNARX structure that learns an inverse model by incorporating prior knowledge of nonlinear friction for linear motor control is proposed in \cite{bolderman2021physics}. In a similar spirit, \cite{fan2023physics} combines a physics-based feedforward with a PGNN-based component for hybrid stepper motor control. PGNN-based feedforward under noisy measurements is further explored in \cite{bolderman2022physics}, while \cite{bolderman2023data} investigates the impact of identification error on tracking error and proposes a control-oriented identification approach. Recently, \cite{lin2026physics} proposes gated recurrent units to learn system inverses, addressing the limitations of NARX-type input-output formulations that arise from latent-state dependency. 
An alternative route for partially measured nonlinear systems is taken in \cite{banerjee2025data}, which develops a data-driven input-output representation and inversion framework for systems exhibiting repetitive motion. This formulation is extended in
\cite{altiner2026datadriven} through a frequency-domain representation that
addresses the dimensionality limitations of the time-domain approach. We adopt
this model inversion framework as the basis of the present work.

Across these approaches, data-driven feedforward design follows a common
two-stage structure: model generation, in which a surrogate of the plant is learned from input-output data, and inverse computation, in which the learned surrogate is inverted through an optimization to obtain the feedforward input. Each stage introduces a distinct source of error that propagates to the realized tracking error. The first is the modeling error between the learned surrogate and the true plant, which depends on the model architecture and the number of data samples. The second is the inversion error introduced when computing the feedforward input from the surrogate, which is typically obtained by solving a nonconvex optimization problem due to the nonlinear nature of the surrogate. While the modeling error can be bounded via statistical tools applied to the surrogate's predictions at unseen inputs, the inversion error is less amenable to certification, as the geometry of complex nonconvex landscapes and exponential scale in input dimension pose a formidable challenge in bounding the gap between a numerically computed inverse and the globally optimal solution, and this gap propagates directly into the tracking error.

To resolve this difficulty, we propose an invertible neural network surrogate
in the frequency domain. Architectural invertibility removes the inversion
error by construction: the inverse map is defined by the model itself and is
evaluated by a contractive fixed-point iteration, whose error decays
geometrically in the number of iterations. Consequently, the feedforward
tracking error is governed by the surrogate modeling error alone, and its
certification reduces to bounding the modeling error at an unseen input,
a problem for which conformal prediction provides distribution-free,
finite-sample guarantees.

The contributions of the paper are summarized as follows:
\begin{itemize}
    \item By adopting an invertible neural network architecture, we bypass solving a nonconvex inversion problem and reduce feedforward control design to a data-driven modeling problem.
    \item Using conformal prediction, we provide probabilistic bounds with finite samples that characterize the feedforward tracking error in terms of the Lipschitz constants of the true map and the surrogate model's inverse, and the modeling error.
    \item We validate the approach and the derived bounds on a mechanical load with nonlinear friction as a numerical example.
\end{itemize}

To the best of our knowledge, architectural invertibility has not previously
been exploited to obtain certified feedforward tracking bounds for unknown
nonlinear systems. Among the works discussed above, the closest is
\cite{bolderman2023data}, which derives a quantitative relation between
tracking error and identification error and designs the identification cost
function to minimize the resulting tracking error. The present work differs in
that the tracking error bound of \cite{bolderman2023data} retains an inversion
error term arising from the optimization-based inverse, whereas the proposed
architecture eliminates this term entirely.


The rest of the paper is organized as follows. Section \ref{sec: sec2} introduces the problem definition. Section \ref{sec: sec3} describes the data-driven feedforward control design framework. The main results are presented in Section \ref{sec: sec4}. Numerical results and conclusion are discussed in Section \ref{sec: sec5} and \ref{sec: sec6}.

\section{Problem Definition}
\label{sec: sec2}
Consider the unknown nonlinear discrete-time system
\begin{subequations}\label{eq: sys_des}
    \begin{equation}\label{system}
        x_{k+1} = f_s(x_k,u_k)
    \end{equation}
    \begin{equation}\label{output}
        y_k = g_s(x_k)
    \end{equation}   
\end{subequations}
where $k\in \mathbb{N}$, $x_k \in \mathbb{R}^{n_x}$ represents the system state vector, $u_k\in \mathbb{R}^{n_u}$ denotes the control input, and $y_k\in \mathbb{R}^{n_y}$ is the system output. The maps $f_s:\mathbb{R}^{n_x}\times \mathbb{R}^{n_u} \rightarrow \mathbb{R}^{n_x}$ and $g_s:\mathbb{R}^{n_x} \rightarrow \mathbb{R}^{n_y}$ define the unknown state transition and output maps, and are assumed to be Lipschitz continuous.

In this context, we assume that both $f_s$ and $g_s$ are entirely unknown, and only partial state measurements are accessible. The available information for the system \eqref{eq: sys_des} is restricted to a collected dataset, consisting of input-output trajectory pairs:
\begin{equation}\label{eq:data}
\mathcal{D} = \{(U_i, Y_i)\}_{i=1}^N, 
\end{equation}
where $U_i\in \mathbb{R}^{n_u h}$ and $Y_i\in \mathbb{R}^{n_y h}$ denote the recorded input and output trajectories over a finite horizon length $h$.

The control objective is to design a data-driven feedforward controller that enables the unknown system to track a specified periodic reference trajectory. This objective is formalized as follows.

\begin{problem}\label{prob: tracking}
Given the unknown nonlinear dynamical system \eqref{eq: sys_des} under partial state measurement and the dataset $\mathcal{D}$ in \eqref{eq:data}, design a feedforward control $u$ such that the system output $y_k$ tracks a given periodic reference trajectory $y_d$ within a bounded error $\varepsilon$, satisfying:
\begin{equation}
    \left\rVert y_k - y_d \right\rVert \leq \varepsilon
\end{equation}
\end{problem}

Data-driven system representations and corresponding model inversion strategies for addressing Problem~\ref{prob: tracking} have been proposed in the time domain~\cite{banerjee2025data} and in the frequency domain~\cite{altiner2026datadriven}. However, a key limitation in these existing frameworks is the lack of a rigorous theoretical characterization of the tracking error $\varepsilon$. Specifically, which critical factors propagate to the open-loop tracking performance remains unaddressed. 

To bridge this gap, this paper investigates the certification for data-driven feedforward control. We demonstrate that by adopting an INN surrogate architecture, the tracking error bound in the frequency domain can be effectively quantified and analyzed through the lens of conformal prediction.

\section{Methodology} 
\label{sec: sec3}
\subsection{Model generation}
Consider the general class of nonlinear systems described in \eqref{eq: sys_des} under partial state measurement, which is a common scenario in practice due to sensing limitations, cost considerations, and physical constraints. Since the system is partially measured, we exploit the history of input-output data to represent the underlying system mapping. Moreover, since the target trajectory is periodic, we adopt the modeling approach of~\cite{banerjee2025data}, which is suited for periodic dynamical systems. Specifically, we utilize the recursive structure of equation~\eqref{eq: sys_des} over the time horizon $h$,

\begin{equation}\label{eq:recursive}
    \begin{bmatrix}
        y_k\\
        y_{k+1}\\
        y_{k+2}\\
        \vdots\\
        y_{k+h}
    \end{bmatrix}= \begin{bmatrix}
        g_s(x_k)\\
        g_s(f_s(x_k,u_k))\\
        g_s(f_s(f_s(x_k,u_k),u_{k+1}))\\
        \vdots\\
        g_s(f_s(\dots f_s(f_s(x_k,u_k),u_{k+1})\dots,u_{k+h-1}))
    \end{bmatrix}
\end{equation}
and for $k=0$, equation~\eqref{eq:recursive} can be expressed in compact form
\begin{equation}\label{eq: msp}
    Y=F(x_0,u_0,\dots, u_{h-1})
\end{equation}
where $Y$ is defined as $Y=\begin{bmatrix}
    y_0&y_1&\dots y_h
\end{bmatrix}^T$. Assuming that the initial conditions are fixed, the compact representation in~\eqref{eq: msp} enables us to cast the dynamic relationship as a static mapping and leverage existing data-driven modeling tools.

A key limitation of this time-domain approach is the curse of dimensionality when the sampling frequency is high. To address this, we transform the available time-domain input--output data into the frequency domain using the Fourier transform. The rationale behind this step is that nonlinear systems driven by periodic inputs generate a discrete set of frequency components in steady state, as established in Volterra series theory. The dataset~\eqref{eq:data} in the frequency domain can be represented as
\begin{equation}
\label{eq: data_freq}
\mathcal{D}_{\omega} = \{(\hat{U}_i(\omega), \hat{Y}_i(\omega))\}_{i=1}^N,
\end{equation}
where $\hat{U}_i(\omega)$ and $\hat{Y}_i(\omega)$ consist of the real and imaginary parts of the complex values at the frequency bins of interest. Specifically, for an excitation centered at frequency $\omega$, the frequency-domain representations of the input and output are constructed as
\begin{equation}
\hat{U}_i(\omega)=\begin{bmatrix}
    \hat{U}_i(0) \\ \Re (\hat{U}_i(\omega)) \\ \Im(\hat{U}_i(\omega))\\ \Re (\hat{U}_i(2\omega)) \\ \Im(\hat{U}_i(2\omega)) \\ \vdots
\end{bmatrix}, \;
\hat{Y}_i(\omega)=\begin{bmatrix}
    \hat{Y}_i(0) \\ \Re (\hat{Y}_i(\omega)) \\ \Im(\hat{Y}_i(\omega))\\ \Re (\hat{Y}_i(2\omega)) \\ \Im(\hat{Y}_i(2\omega)) \\ \vdots
\end{bmatrix}
\end{equation}
Thus, the representation of an unknown system through the dataset~\eqref{eq: data_freq} reduces to learning a mapping such that
\begin{equation}
   \hat{Y}(\omega) = \hat{F}(\hat{U}(\omega))
\end{equation}
where $\hat{F}$ is learned within a nonlinear model class $\mathcal{F}$ by solving
\begin{equation}
\label{eq: modeling}
  \hat{F} \in \operatorname*{arg\,min}_{\hat{F}\in \mathcal{F}} J_r(\hat{F})
\end{equation}
where $J_r(\hat{F})$ denotes the regression loss. 

In prior work \cite{altiner2026datadriven}, the learned map $F$ obtained from~\eqref{eq: modeling} is inverted by solving
\begin{equation}
    \hat{U}^*(\omega) = \operatorname*{arg\,min}_{U} \left\lVert \hat{F}(\hat{U}(\omega)) - \hat{Y}_d(\omega) \right\rVert
\end{equation}
to obtain the feedforward input. However, since $\hat{F}$ is a nonlinear surrogate, this inversion problem is generally nonconvex, and the gap between local and global minimum introduces an inversion error that constitutes an additional source of tracking degradation. To address this, we restrict $\hat{F}$ to the invertible residual network 
(i-ResNet) architecture of~\cite{behrmann2019invertible}, which guarantees an exact inverse via fixed-point iteration and enables us to use statistical methods to bound tracking error with finite data samples.

\subsection{Invertible Neural Networks} \label{sec: 3b}
The i-ResNet architecture introduced by~\cite{behrmann2019invertible} is 
given by
\begin{equation}
\label{eq: INN}
    \hat{Y} = \hat{U} + c\, G_{\theta}(\hat{U})
\end{equation}
where $G_{\theta}$ is a multi-layer feedforward neural network with $\theta$ represents weights and bias to be tuned and $c \in (0,1)$ is a fixed scalar. The map \eqref{eq: INN} defines a surrogate $\hat{F}(\hat{U}) = \hat{U} + 
c\,G_\theta(\hat{U})$. Invertibility of $\hat{F}$ is guaranteed by the contraction condition $c\, L_{G_\theta} < 1$, which is enforced during training by applying spectral normalization to each weight matrix, ensuring $L_{G_\theta} \leq 1$, combined with the fixed choice $c < 1$.

By the Banach fixed-point theorem, a contractive map admits a unique fixed point, and this property is exploited to compute the inverse $\hat{F}^{-1}$ 
via fixed-point iteration, which converges geometrically to the true inverse.

This architectural choice enables feedforward input computation in the frequency domain without solving any nonconvex optimization problem, reducing tracking error certification to a surrogate modeling problem addressable via conformal prediction.

\section{Main Result}
\label{sec: sec4}
In this paper, we exploit this architectural invertibility to compute the feedforward control input. Whereas a generic learned surrogate would require solving a nonconvex optimization problem to invert, the i-ResNet defines its inverse map as part of the architecture itself, reducing inversion to a fixed-point iteration. The following lemma provides the characterization of the tracking error in terms of the Lipschitz constants of the unknown plant and the inverse neural networks.

\begin{lemma}
\label{lem: tracking}
Let $F$ denote the true frequency-domain map, $\hat{F}$ its i-ResNet 
approximation, $\hat{U}_{inn} = \hat{F}^{-1}(\hat{Y}_d)$ the feedforward input obtained 
by inverting $\hat{F}$, and $\hat{U}^*$ the ideal feedforward input satisfying 
$F(U^*) = Y_d$. Then the tracking error satisfies
\begin{equation}
\label{eq: bound}
    \left\lVert \hat{Y} - \hat{Y}_d \right\rVert \leq L_F \, L_{\hat{F}^{-1}} \,
    \left\lVert F(\hat{U}^*) - \hat{F}(\hat{U}^*) \right\rVert
\end{equation}
where $L_F$ is the Lipschitz constant of $F$, $L_{\hat{F}^{-1}}$ is the Lipschitz constant of $\hat{F}^{-1}$, and $\hat{Y}_d$ is the reference signal in the frequency domain.
\end{lemma}

\begin{proof}
Using $Y = F(\hat{U}_{inn})$ and $\hat{Y}_d = F(\hat{U}^*)$, and applying the Lipschitz 
continuity of $F$ and $\hat{F}^{-1}$:
\begin{align*}
    \left\lVert \hat{Y} - \hat{Y}_d \right\rVert 
    &= \left\lVert F(\hat{U}_{inn}) - F(\hat{U}^*) \right\rVert \\
    &\leq L_F \left\lVert \hat{U}_{inn} - \hat{U}^* \right\rVert \\
    &= L_F \left\lVert \hat{F}^{-1}(\hat{Y}_d) - \hat{F}^{-1}(\hat{F}(\hat{U}^*)) 
    \right\rVert \\
    &\leq L_F \, L_{\hat{F}^{-1}} \left\lVert \hat{Y}_d - \hat{F}(\hat{U}^*) 
    \right\rVert \\
    &= L_F \, L_{\hat{F}^{-1}} \left\lVert F(\hat{U}^*) - \hat{F}(\hat{U}^*) 
    \right\rVert
\end{align*}
where the last equality uses $\hat{Y}_d = F(\hat{U}^*)$.
\end{proof}
In Lemma~\ref{lem: tracking}, we show that the frequency-domain tracking error can be bounded by the modeling error of the learned surrogate evaluated at the ideal feedforward input $U^*$, scaled by the Lipschitz constants of the true frequency-domain map and the inverse of the learned surrogate. 

In practice, $\hat{U}^*$ is generally not contained in the finite training dataset and can be regarded as an unseen input. Therefore, the ability of the trained surrogate to generalize to unseen frequency-domain inputs is the key factor determining the certified tracking performance. Since $\hat{U}^*$ is unknown, the right-hand side of \eqref{eq: bound} is not directly computable from data. To obtain a computable certificate, we replace the modeling error evaluated at $\hat{U}^*$ with a bound obtained by conformal prediction. Conformal prediction is a distribution-free statistical framework that provides finite-sample guarantees and uncertainty quantification on the predictions of a learned model at unseen inputs, without assumptions on the underlying data distribution \cite{angelopoulos2023conformal}. It has recently attracted attention in the control community as a tool for certifying learning-enabled components in control systems \cite{lindemann2025formal}. 

Before stating the main result and to implement conformal prediction (CP), we partition the dataset $\mathcal{D}_\omega$ into a training set $\mathcal{D}^{train}_{\omega} = \{(\hat{U}_i(\omega), 
\hat{Y}_i(\omega))\}_{i=1}^{N_{train}}$ and a calibration set 
$\mathcal{D}^{cal}_{\omega} = \{(\hat{U}_i(\omega), 
\hat{Y}_i(\omega))\}_{i=1}^{N_{cal}}$, where $\mathcal{D}^{train}_{\omega}$ is used to train the i-ResNet surrogate and $\mathcal{D}^{cal}_{\omega}$ 
is used to compute the nonconformity scores for the conformal predictor. This specific type of CP is called Split CP. We define the nonconformity score as
\begin{equation}
\label{eq: nonconformity}
    \mu_{i} = \left\lVert \hat{Y}_i - \hat{F}(\hat{U}_i) \right\rVert, 
    \quad i = 1, \dots, N_{cal}
\end{equation}
where the scores are computed on $\mathcal{D}^{cal}_\omega$. With these definitions, we state the following assumptions, which are required for the main result.
\begin{assumption}
The dataset $\mathcal{D}_\omega$ consists of independent and identically distributed samples from an underlying distribution, and the test pair $(\hat{U}^*, \hat{Y}_d)$ is an independent draw from the same distribution.
\end{assumption}

Although Assumption~1 is the main prerequisite for CP, it is not restrictive in practice. The reference trajectory $Y_d$ lies 
within the support of the output distribution provided that the dataset is constructed from periodic signals of the same class as the reference, for example, sums of sinusoids. Moreover, the ideal feedforward control input $U^*$ lies within the support of the input distribution as well, provided that data collection is carried out under the same operating conditions as deployment. In other words, Assumption~1 requires that the reference trajectory and the corresponding ideal feedforward input are not extrapolated beyond the support of the training distribution. Thus, given a failure probability $\delta \in (0,1)$, we have the following result.

\begin{theorem}
Under Assumption~1 and Lemma~\ref{lem: tracking}, the feedforward tracking 
error in the frequency domain satisfies
\begin{equation}
\label{eq: bound_CP}
    \mathrm{Prob}\left( \left\lVert \hat{Y} - \hat{Y}_d \right\rVert \leq 
    L_F \, L_{\hat{F}^{-1}} \, q_{1-\delta} \right) \geq 1-\delta
\end{equation}
where $q_{1-\delta} = \mu_{(\lceil (N_{test}+1)(1-\delta) \rceil)}$ is the 
$\lceil (N_{test}+1)(1-\delta) \rceil$-th quantile of the 
nonconformity scores $\{\mu_i\}_{i=1}^{N_{test}}$.
\end{theorem}

\begin{proof}
From the result in Lemma 1 and under Assumption 1, the prediction of $\left\lVert F(\hat{U}^*) - \hat{F}(\hat{U}^*) 
    \right\rVert$ is guaranteed at most $q_{1-\delta} = \mu_{(\lceil (N_{test}+1)(1-\delta) \rceil)}$ with the probability of $1-\delta$ by Theorem D.1 in \cite{angelopoulos2023conformal}. Substituting this into the bound in Lemma 1 yields the bound in \eqref{eq: bound_CP}.  
\end{proof}

The physical interpretation of Theorem 1 is as follows. For a system operating within a fixed regime, any reference signal whose corresponding feedforward input is drawn i.i.d. from the same distribution as the calibration data, that is, drawn from the same operating distribution, satisfies the certified tracking bound with probability at least $1-\delta$. Moreover, $q_{1-\delta}$ serves as a probabilistic upper bound on the modeling error at the unseen point $U^*$: with probability at least $1-\delta$, the surrogate prediction error $\left\lVert F(\hat{U}^*) - \hat{F}(\hat{U}^*) \right\rVert$ does not exceed $q_{1-\delta}$. 

Although $q_{1-\delta}$ can be computed directly from the nonconformity scores, the Lipschitz constants $L_F$ and $L_{\hat{F}^{-1}}$ must also be determined to evaluate the bound in \eqref{eq: bound}. For the inverse surrogate, the global Lipschitz constant introduces excessive conservatism. To mitigate this, we propose using a local Lipschitz constant is estimated at the $q_{1-\delta}$-neighbourhood of the reference trajectory. To estimate the Lipschitz constant $L_F$ of the unknown true map, the estimation problem can be formulated as an uncertain optimization problem, and the available data can be utilized to solve it as a scenario convex program \cite{campi2008exact}. This choice introduces another level of failure rate with a confidence level, which must be reflected in the probabilistic bound in \eqref{eq: bound_CP}. In this paper, we omit this analysis due to page restrictions; therefore, it is assumed to be known.

\section{Numerical Results}
\label{sec: sec5}

Consider the following second-order mechanical system
\begin{align}
    \label{eq: DC1}
    \dot{x}_1 &= x_2\\
    \dot{x}_2 &= a_1u - a_1 F_{fric}(x_2) - a_2 x_1\\
    y &= x_1
    \label{eq: DC2}
\end{align}
where $x_1$ and $x_2$ denote position and velocity, respectively, and $u$ represents the control input. The term $F_{fric}(x_2)$ models smooth nonlinear friction and is given by
\begin{align}
\label{eq: friction}
    F_{fric}(x_2) &= \alpha_1\tanh (\beta_1x_2) \\ &+ \alpha_2 \big[\tanh(\beta_2 x_2)-\tanh(\beta_3 x_2)\big] + \alpha_3 x_2\notag
\end{align}
where the term $\alpha_1\tanh (\beta_1x_2)$ represents the Coulomb friction, the second term captures the Stribeck effect and $\alpha_3 x_2$ is the viscous friction \cite{yao2015output}. The parameters are given in Table \ref{tab: Mechparam}.
\begin{table}
    \centering
    \begin{tabular}{|c|c|c|}
    \hline
    Parameters & Value & Units \\ 
    \hline
         $a_1$&  400 & $rad/(V\cdot s^{-2})$ \\
         $a_2$&  355 & $s^{-2}$ \\
         $\alpha_1$&  0.04 & $V$\\
         $\alpha_2$&  0.01 & $V$\\
         $\alpha_3$& 0.05 & $V\cdot s/rad$\\
         $\beta_1$& 15 & $s/rad$\\
         $\beta_2$& 15 & $s/rad$\\
         $\beta_3$& 1.5 & $s/rad$\\
         \hline
    \end{tabular}
    \caption{Parameters used in the mechanical system and the friction term}
    \label{tab: Mechparam}
\end{table}

The surrogate follows the i-ResNet architecture of Section~\ref{sec: 3b}, with six cascaded blocks of 128-neuron single-layer networks and exponential linear unit activations. The training data are generated by simulating the model \eqref{eq: DC1} - \eqref{eq: DC2} under input signals composed of three sinusoidal components selected from $\{0, 1, 2, 3, 4, 5\}$ Hz, where 0 Hz denotes the DC term. Parameterizing each input by three components reflects a key property of the
framework: the harmonic content required for accurate tracking is not assumed known a priori, so the data cover all combinations of active components and the inverse model is expected to select the appropriate combination for a given reference. For each combination, amplitudes are sampled from $(0,\; 0.10)$ for the DC term and $(0,\; 0.35)$ for each active sinusoidal component, with phases sampled from $(0,\; 2\pi)$. A total of 1.28 million data points is collected. The total data set is partitioned into three subsets as $70\%$ for training data, $15\%$ for validation used for early stopping, and $15\%$ for the calibration data set which is held out from both training and validation and used in the computation of the nonconformity scores in \eqref{eq: nonconformity}. The input-output dimension of the neural network is 11, comprising the DC component and the real and imaginary parts at 1, 2, 3, 4, and 5 Hz. Training is performed in PyTorch using the Adam optimizer with 1000 epochs at $10^{-3}$ learning rate. The trained model achieves a RMSE of $0.09$ deg on the train and the test data set. 

To test the tracking performance, we compute the inverse through the architecture via the Fixed-point algorithm with 200 iterations, the reference signal $y_d=15\sin(2\pi t)$ deg. The nonlinearities of the system involve the $\tanh$ function what generates an infinite number of odd harmonics. Since the surrogate is defined over frequency components up to 5 Hz, the feedforward input is truncated accordingly. Figure \ref{fig: time_resp_mech} - \ref{fig: control_input} shows the time-domain tracking performance, the tracking error in frequency domain and the feedforward control input in the frequency domain. The time domain response clearly shows that the feedforward input obtain from the inversion of the surrogate model achieves tracking successfully. To assess the success of the control, we report that the maximum tracking error is 0.495 deg, while the analytically computed error from the truncated signal is 0.462 deg. The frequency domain feedforward input clearly shows that the feedforward input includes harmonics at 3 and 5 Hz which are required to be included in the feedforward control to compensate the nonlinearities of the system.  

\begin{figure}
    \centering
    \includegraphics[width=0.9\linewidth]{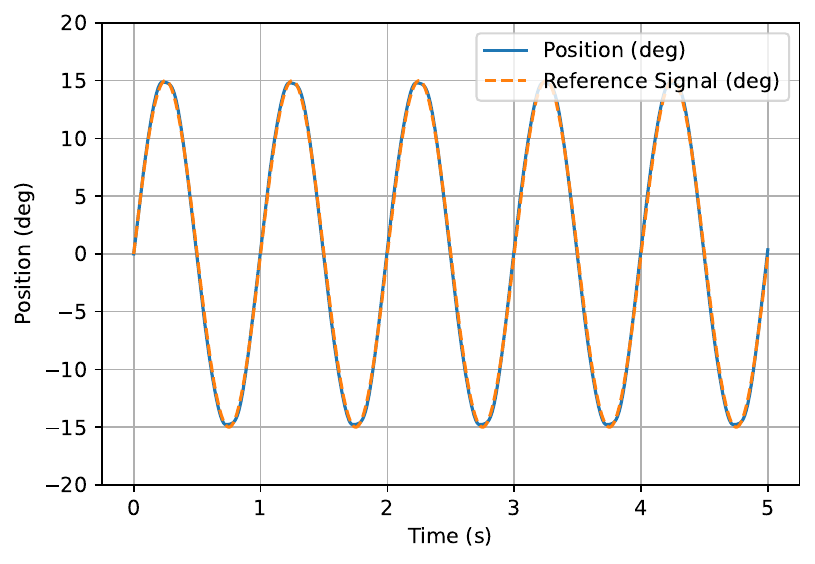}
    \caption{Reference and system response using the data-driven control input in time domain}
    \label{fig: time_resp_mech}
\end{figure}

\begin{figure}
    \centering
    \includegraphics[width=0.9\linewidth]{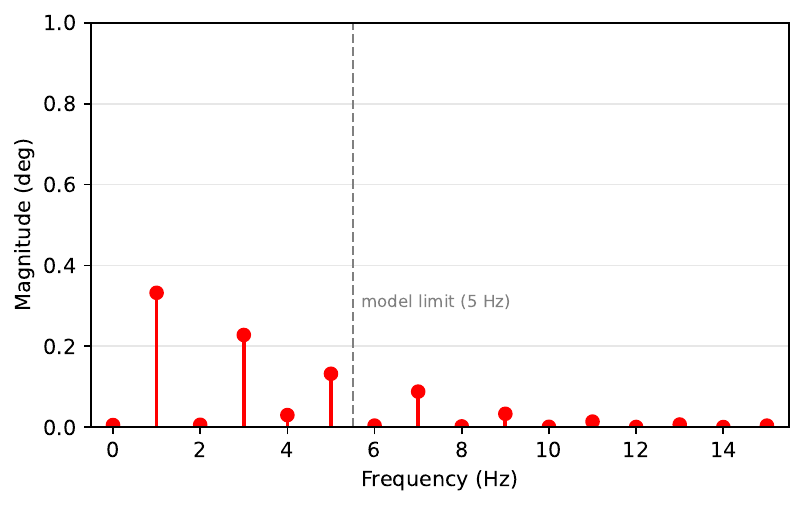}
    \caption{Tracking error in frequency domain}
    \label{fig: error_fft}
\end{figure}

\begin{figure}
    \centering
    \includegraphics[width=0.9\linewidth]{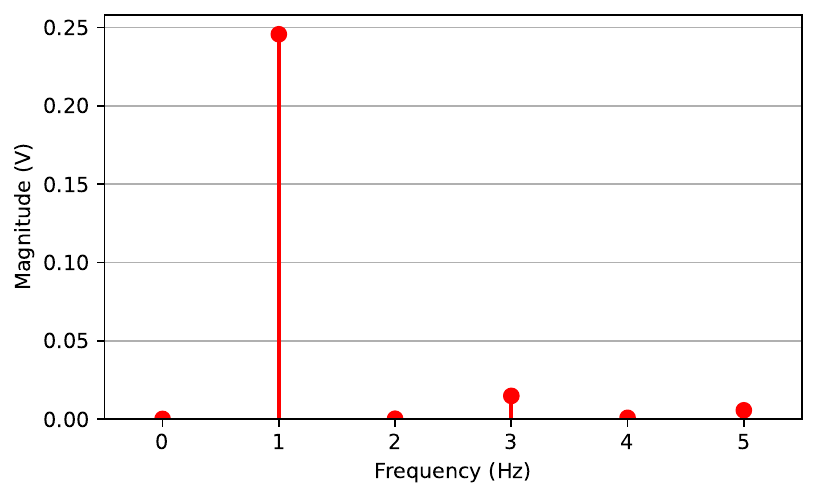}
    \caption{Feedforward control input in frequency domain}
    \label{fig: control_input}
\end{figure}

We implement Theorem 1 to show the validity of the bound \eqref{eq: bound_CP}. As stated in Theorem 1, we need to calculate the quantile $q_{1-\delta}$, the Lipschitz constant of the true map and the inverse surrogate model's local Lipschitz constant for the given reference. The quantile can be easily obtained by sorting the nonconformity scores \eqref{eq: nonconformity} and taking the
$\lceil(N_{test}+1)(1-\delta)\rceil$-th order statistic. For a failure rate $\delta=0.1$, Figure \ref{fig: nonconformity} demonstrates the distribution of the nonconformity scores. The red line in Figure \ref{fig: nonconformity} marks the quantile $90\%$ ($q_{1-\delta}$) of the nonconformity scores and its value is 0.446 deg. It means that the modeling error for a new data sample from the same distribution with the test data set is at most 0.446 deg with the probability of $90\%$. The local Lipschitz constant of the true map $L_F$ and the inverse surrogate are estimated as 1.52 and 3.28, respectively. Since the certified bound in \eqref{eq: bound_CP} is expressed as a
frequency-domain norm, it is mapped to a peak time-domain error through the relation $\max_t |e(t)| \le \sqrt{6}\,\lVert e \rVert_2$, following from the Cauchy-Schwarz inequality over the six retained frequency components. The tracking error for the reference signal $y_d=15\sin(2\pi t)$ deg is demonstrated in Figure \ref{fig: time_error} and it certifies that the bound \eqref{eq: bound_CP} holds. Figure~\ref{fig: boundvsref} shows that the bound \eqref{eq: bound_CP} holds across reference signals of different amplitudes, empirically supporting the theoretical analysis. The bound is looser at lower amplitudes because the velocity remains in the low-speed regime, where friction dominates and the forward map is least sensitive to the input. The inverse Lipschitz constant $L_{\hat{F}^{-1}}$, which we observe to decrease monotonically with amplitude over the tested range, amplifies the certified modeling error; the bound is therefore more conservative precisely where the nonlinearity is strongest.

\begin{figure}
    \centering
    \includegraphics[width=0.9\linewidth]{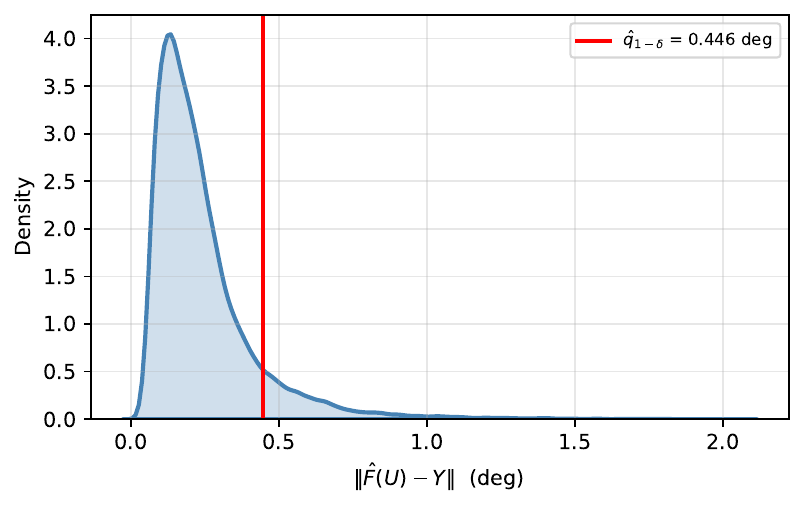}
    \caption{The histogram of the nonconformity scores}
    \label{fig: nonconformity}
\end{figure}

\begin{figure}
    \centering
    \includegraphics[width=0.9\linewidth]{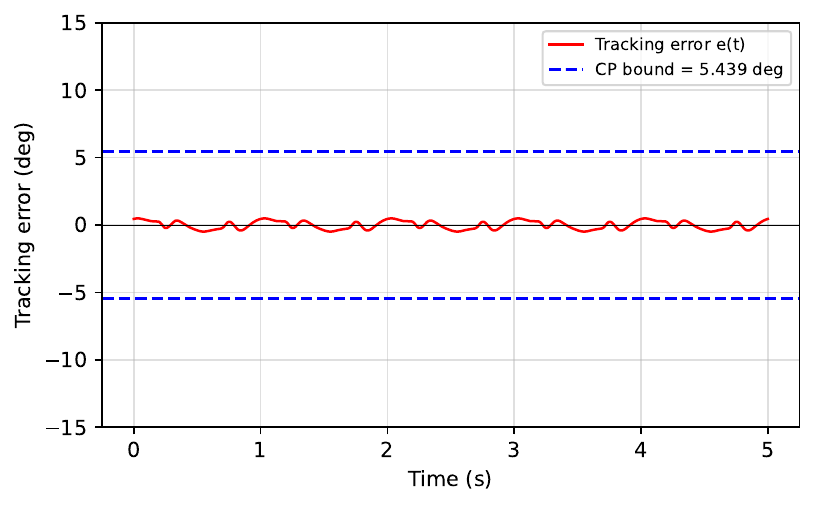}
    \caption{Tracking error and CP bounds in time domain}
    \label{fig: time_error}
\end{figure}

\begin{figure}
    \centering
    \includegraphics[width=0.9\linewidth]{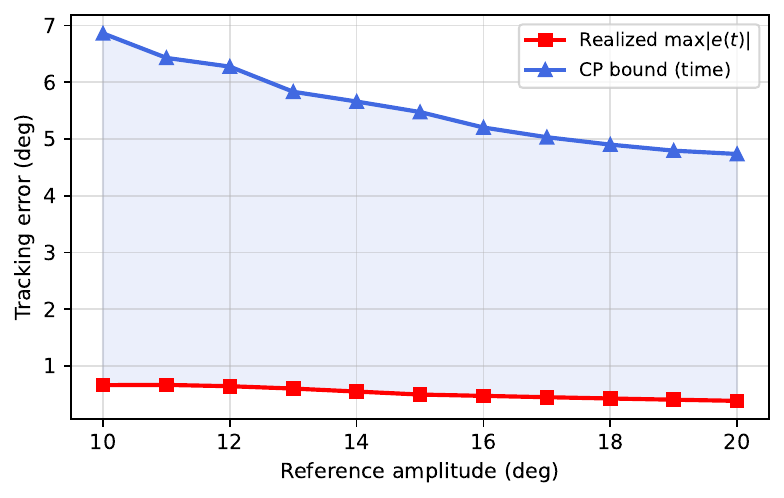}
    \caption{The bound vs. reference amplitude}
    \label{fig: boundvsref}
\end{figure}



\section{Conclusion}
\label{sec: sec6}
In this paper, we investigated the certification of data-driven model inversion for feedforward control. To address the difficulty of characterizing the error arising from nonconvex inversion of a learned surrogate, we adopted i-ResNets as the model architecture, exploiting
architectural invertibility to compute the feedforward input in the frequency domain without solving a nonconvex optimization problem. This choice reduces tracking error certification to a surrogate modeling
problem. Combining Lipschitz continuity of the true plant and the inverse surrogate with conformal prediction, we derived probabilistic guarantees on the feedforward tracking error, validated on a DC-motor-driven mechanical load with nonlinear friction.

\section*{Acknowledgements}
This work was funded by DEVCOM Army Research Laboratory as a part of the VICTOR ERP under Cooperative Agreement no. W911NF-20-2-0161 with Dr. Mike Kweon as the Program Manager and Dr. Jacob Temme as the Deputy Program Manager. The views and conclusions contained in this document are those of the authors and should not be interpreted as representing the official policies, either expressed or implied, of DEVCOM Army Research Laboratory or the U.S. Government. The U.S. Government is authorized to reproduce and distribute reprints for Government purposes notwithstanding any copyright notation herein.

\ifCLASSOPTIONcaptionsoff
  \newpage
\fi



%
\bibliographystyle{IEEEtran}
\bibliography{refs_inn}

@inproceedings{behrmann2019invertible,
  title={Invertible residual networks},
  author={Behrmann, Jens and Grathwohl, Will and Chen, Ricky TQ and Duvenaud, David and Jacobsen, J{\"o}rn-Henrik},
  booktitle={International conference on machine learning},
  pages={573--582},
  year={2019},
  organization={PMLR}
}

@article{banerjee2025data,
  title={Data-driven modeling and control framework under partial state measurements with experimental validation on multi-fuel engines},
  author={Banerjee, Arunava and Sarkar, Rajasree and Altiner, Ihsan Berk and Govind Raju, Sathya Aswath and Sun, Zongxuan and Kim, Kenneth and Kweon, Chol-Bum Mike},
  journal={Proceedings of the Institution of Mechanical Engineers, Part I: Journal of Systems and Control Engineering},
  pages={09596518251399937},
  year={2025},
  publisher={SAGE Publications Sage UK: London, England}
}

@inproceedings{bolderman2023data,
  title={Data-driven feedforward control design for nonlinear systems: A control-oriented system identification approach},
  author={Bolderman, Max and Lazar, Mircea and Butler, Hans},
  booktitle={2023 62nd IEEE Conference on Decision and Control (CDC)},
  pages={4530--4535},
  year={2023},
  organization={IEEE}
}

@article{lindemann2025formal,
  title={Formal verification and control with conformal prediction: Practical safety guarantees for autonomous systems},
  author={Lindemann, Lars and Zhao, Yiqi and Yu, Xinyi and Pappas, George J and Deshmukh, Jyotirmoy V},
  journal={IEEE Control Systems},
  volume={45},
  number={6},
  pages={72--122},
  year={2025},
  publisher={IEEE}
}

@article{angelopoulos2023conformal,
  title={Conformal prediction: A gentle introduction},
  author={Angelopoulos, Anastasios N and Bates, Stephen},
  journal={Foundations and Trends in Machine Learning},
  volume={16},
  number={4},
  pages={494--591},
  year={2023},
  publisher={Emerald Publishing Limited}
}

@article{campi2008exact,
  title={The exact feasibility of randomized solutions of uncertain convex programs},
  author={Campi, Marco C and Garatti, Simone},
  journal={SIAM Journal on Optimization},
  volume={19},
  number={3},
  pages={1211--1230},
  year={2008},
  publisher={SIAM}
}

@article{altiner2026datadriven,
  author  = {Altiner, Berk and Banerjee, Arunava and Sarkar, Rajasree and Sun, Zongxuan and Kim, Kenneth and Kweon, Chol-Bum Mike},
  title   = {Data-driven Trajectory Tracking under Partial State Measurements: A Frequency Domain Approach},
  journal = {International Journal of Robust and Nonlinear Control},
  year    = {2026},
  note    = {Submitted for publication},
}

@inproceedings{bolderman2021physics,
  title={Physics--guided neural networks for inversion--based feedforward control applied to linear motors},
  author={Bolderman, Max and Lazar, Mircea and Butler, Hans},
  booktitle={2021 IEEE Conference on Control Technology and Applications (CCTA)},
  pages={1115--1120},
  year={2021},
  organization={IEEE}
}

@inproceedings{fan2023physics,
  title={Physics--guided neural networks for inversion--based feedforward control applied to hybrid stepper motors},
  author={Fan, Daiwei and Bolderman, Max and Koekebakker, Sjirk and Butler, Hans and Lazar, Mircea},
  booktitle={2023 IEEE Conference on Control Technology and Applications (CCTA)},
  pages={1153--1158},
  year={2023},
  organization={IEEE}
}

@inproceedings{zhang2007inverse,
  title={Inverse model identification of nonlinear dynamic system using neural network},
  author={Zhang, Ming-Guang},
  booktitle={2007 International Conference on Machine Learning and Cybernetics},
  volume={5},
  pages={2451--2455},
  year={2007},
  organization={IEEE}
}

@inproceedings{wanigasekara2019neural,
  title={Neural network based inverse system identification from small data sets},
  author={Wanigasekara, Chathura and Swain, Akshya and Nguang, Sing Kiong and Prusty, B Gangadhara},
  booktitle={2019 International Joint Conference on Neural Networks (IJCNN)},
  pages={1--6},
  year={2019},
  organization={IEEE}
}

@article{aarnoudse2024control,
  title={Control-relevant neural networks for feedforward control with preview: Applied to an industrial flatbed printer},
  author={Aarnoudse, Leontine and Kon, Johan and Ohnishi, Wataru and Poot, Maurice and Tacx, Paul and Strijbosch, Nard and Oomen, Tom},
  journal={IFAC Journal of Systems and Control},
  volume={27},
  pages={100241},
  year={2024},
  publisher={Elsevier}
}

@inproceedings{bolderman2022physics,
  title={Physics--guided neural networks for feedforward control: From consistent identification to feedforward controller design},
  author={Bolderman, Max and Lazar, Mircea and Butler, Hans},
  booktitle={2022 IEEE 61st Conference on Decision and Control (CDC)},
  pages={1497--1498},
  year={2022},
  organization={IEEE}
}

@inproceedings{aarnoudse2021control,
  title={Control-relevant neural networks for intelligent motion feedforward},
  author={Aarnoudse, Leontine and Ohnishi, Wataru and Poot, Maurice and Tacx, Paul and Strijbosch, Nard and Oomen, Tom},
  booktitle={2021 IEEE International Conference on Mechatronics (ICM)},
  pages={1--6},
  year={2021},
  organization={IEEE}
}

@article{devasia2002should,
  title={Should model-based inverse inputs be used as feedforward under plant uncertainty?},
  author={Devasia, Santosh},
  journal={IEEE Transactions on automatic control},
  volume={47},
  number={11},
  pages={1865--1871},
  year={2002},
  publisher={IEEE}
}

@inproceedings{boerlage2003model,
  title={Model-based feedforward for motion systems},
  author={Boerlage, Matthijs and Steinbuch, Maarten and Lambrechts, Paul and Van De Wal, Marc},
  booktitle={Proceedings of 2003 IEEE Conference on Control Applications, 2003. CCA 2003.},
  volume={2},
  pages={1158--1163},
  year={2003},
  organization={IEEE}
}

@article{jamaludin2009friction,
  title={Friction compensation of an $ XY $ feed table using friction-model-based feedforward and an inverse-model-based disturbance observer},
  author={Jamaludin, Zamberi and Van Brussel, Hendrik and Swevers, Jan},
  journal={IEEE Transactions on Industrial Electronics},
  volume={56},
  number={10},
  pages={3848--3853},
  year={2009},
  publisher={IEEE}
}

@article{lin2026physics,
  title={Physics-guided gated recurrent units for inversion-based feedforward control},
  author={Lin, Mingdao and Bolderman, Max and Lazar, Mircea},
  journal={IEEE Transactions on Control Systems Technology},
  year={2026},
  publisher={IEEE}
}

@article{xu2024robust,
  title={Robust inversion-based feedforward control with hybrid modeling for feed drives},
  author={Xu, Haijia and Hinze, Christoph and Iannelli, Andrea and Verl, Alexander},
  journal={IEEE Transactions on Control Systems Technology},
  volume={33},
  number={3},
  pages={858--871},
  year={2024},
  publisher={IEEE}
}

@inproceedings{kon2022physics,
  title={Physics-guided neural networks for feedforward control: An orthogonal projection-based approach},
  author={Kon, Johan and Bruijnen, Dennis and Van De Wijdeven, Jeroen and Heertjes, Marcel and Oomen, Tom},
  booktitle={2022 American Control Conference (ACC)},
  pages={4377--4382},
  year={2022},
  organization={IEEE}
}

@article{yao2015output,
  title={Output feedback robust control of direct current motors with nonlinear friction compensation and disturbance rejection},
  author={Yao, Jianyong and Jiao, Zongxia and Ma, Dawei},
  journal={Journal of Dynamic Systems, Measurement, and Control},
  volume={137},
  number={4},
  pages={041004},
  year={2015},
  publisher={American Society of Mechanical Engineers}
}

\end{document}